\documentclass[11pt]{article}
\usepackage[utf8]{inputenc}

\usepackage{amsmath, amsthm}
\usepackage{thmtools}
\usepackage{thm-restate}

\usepackage[margin=3.5cm]{geometry}

\usepackage{amssymb}

\newtheorem{theorem}{Theorem}

\newtheorem{lemma}[theorem]{Lemma}
\newtheorem{claim}[theorem]{Claim}
\newtheorem{definition}[theorem]{Definition}
\newtheorem{proposition}[theorem]{Proposition}

\newtheorem{question}[theorem]{Question}
\newtheorem{assumption}[theorem]{Assumption}

\newcommand{\items}{\mathcal{M}} 
\newcommand{\agents}{\mathcal{N}}
\newcommand{\otilde}{\tilde{O}}

\title{On MMS allocations with few items}

\author{Uriel Feige\thanks{Weizmann Institute, Israel. {\tt uriel.feige@weizmann.ac.il}}}

\begin{document}

\maketitle

\begin{abstract}
    We consider the problem of allocating $m$ indivisible items to $n$ agents. For $n \ge 2$, we define $\mu(n)$ as the largest $m$ for which every allocation instance with $n$ agents and $m$ items has an allocation that gives each agent at least her maximin share (MMS). We prove that for general monotone valuations, $\mu(n) = n + \Theta(\log n)$. This holds both for goods and for chores. For arbitrary valuations (not monotone), we prove that $\mu(n) = \log n + O(1)$. 
\end{abstract}




\section{Introduction}

We consider the problem of allocating a set $\items$ of $m \ge 1$ indivisible items to a set $\agents$ of $n \ge 2$ agents. Each agent $i$ has a valuation function $v_i: 2^{\items} \rightarrow \cal{R}$, where $v_i(\emptyset) = 0$. An allocation $A = (A_1, \ldots, A_n)$ is a partition of $\items$ into $n$ disjoint sets, where for each $i \in [n]$, set $A_i$ is given to agent $i$. We wish the allocation to be ``fair". There are several notions of fairness, some of which are {\em feasible} (an allocation satisfying the fairness notion exists in every allocation instance), and some are not (some allocation instances do not have an allocation satisfying the fairness notion). For fairness notions that are not feasible, it is desirable to identify restrictions on the allocation instance that make it feasible. Our work contributes to this line of research.

We are interested in a fairness notion that is known as the {\em maximin share} (MMS)~\cite{budish2011combinatorial}, defined as:

$$MMS(\items,v_i,n) = \max_{P = (P_1, \ldots, P_n)}\min_j v_i(P_j)$$

Here, $P$ ranges over all partitions of $\items$. A maximizing partition $P$ is referred to as an {\em MMS partition}. {A bundle $B$ is {\em acceptable} for agent $i$ if $v_i(B) \ge MMS(\items,v_i,n)$. An {\em MMS allocation} is an allocation $A = (A_1, \ldots, A_n)$ in which for every agent $i$, bundle $A_i$ is acceptable for agent $i$.}

MMS allocations are not guaranteed to exist. A well known negative example has two agents and four items, $e_{11}, e_{12}, e_{21}, e_{22}$ arranged in a two by two matrix, see Table~\ref{tab:negative}. 
\begin{table}[hbt!]
        \centering
\begin{tabular}{|l|l|}
 \hline
 $e_{11}$                  & $e_{12}$  \\
 \hline
$e_{21}$               &   $e_{22}$ \\
 \hline
\end{tabular}
        \caption{Every non-empty bundle of items has value~1, except that for one agent, any bundle that contains two items from the same row has value~2, and for the other agent,  any bundle that contains two items from the same column has value~2.}
        \label{tab:negative}
\end{table}
The MMS partition of the first agent is into the two rows and of the other agent is into the two columns. For both agents, the MMS value is~2, but in every allocation, some agent gets a value of at most~1. 

Define $\mu(n)$ as a function that for each value of $n$ returns the largest value of $m$ for which every allocation instance with $n$ agents and at most $m$ items has an MMS allocation. In this paper, we study the asymptotic growth rate of $\mu(n)$. We do so for three general classes of valuation functions.

\begin{itemize}   
    \item {\em Monotone valuations over goods.} The term {\em goods} refers to items that have non-negative marginal value. Valuations over goods are monotone, in the sense that $v_i(T) \ge v_i(S)$ for every $S \subset T$. {Every bundle in the MMS partition of agent $i$ is acceptable for $i$. For goods, also every bundle that contains a bundle from $i$'s MMS partition is acceptable for $i$. Hence,} the worst case to consider is when every agent $i$ has a unique MMS partition $P^i = (P^i_1, \ldots, P^i_n)$, and the only bundles acceptable for $i$ are those that contain some $P^i_j$ in the MMS partition.
    All parts of the MMS partition are non-empty (as otherwise the MMS is~0, and the agent can be discarded without receiving any item). Consequently, $\mu(n) \ge n$. An MMS allocation is one in which each agent gets a set that contains (or is equal to) one of the parts in her MMS partition.

    \item {\em Monotone valuations over chores.} The term {\em chores} refers to items that have non-positive marginal value. Valuations over chores are monotone, in the sense that $v_i(T) \le v_i(S)$ for every $S \subset T$. In this case, the worst case to consider is when every agent $i$ has a unique MMS partition $P^i = (P^i_1, \ldots, P^i_n)$, 
    and the only bundles acceptable for $i$ are those that are contained in some $P^i_j$ in the MMS partition. 
    Also for chores, $\mu(n) \ge n$, because when $m \le n$, giving each agent at most one chore is an MMS allocation. For $m \ge n$, we may assume that all parts of the MMS partition are non-empty, as otherwise, a chore can be moved from a part that contains more than one chore into an empty part, without decreasing the value of the worst part. An MMS allocation is one in which each agent gets a set that is contained in (or is equal to) one of the parts in her MMS partition.

    \item {\em Arbitrary valuations}. In this case, the worst case to consider is when every agent $i$ has a unique MMS partition $P^i = (P^i_1, \ldots, P^i_n)$, with $v_i(P^i_j) = 0$ for every $j\in [n]$, and $v_i(S) < 0$ for every other non-empty set. Importantly, some parts of the MMS partition may be empty (and will be empty, as it will turn out that $\mu(n) < n$ for sufficiently large $n$). An MMS allocation is one in which each agent gets one of the parts in her MMS partition (that part may be empty).

\end{itemize}

The value of $\mu(n)$ (with different terminology) was previously studied in the special case of additive valuations (valuations satisfying $v(S) = \sum_{e\in S} v(\{e\})$). When $n=2$, MMS allocations always exist, by the {\em cut and choose procedure.} However, for every $n \ge 3$, there are examples showing that MMS allocations do not exist~\cite{kurokawa2018fair}. Both for additive valuations over goods and for additive valuations over chores it was shown that $\mu(n) \ge n + \Omega(\frac{\log n}{\log\log n})$~\cite{hummel2023lowerboundsmaximinshare}. 

Our main result is the following.

\begin{theorem}
\label{thm:main}
    Asymptotically, as $n$ grows, the following bounds hold. 
    \begin{enumerate}
        \item For monotone valuations over goods, $\mu(n) = n + \Theta(\log n)$.
        \item For monotone valuations over chores, $\mu(n) = n + \Theta(\log n)$.
        \item For arbitrary valuations, $\log n - O(1) \le \mu(n) \le \log n + O(1)$. 
    \end{enumerate} 
\end{theorem}

The constants hidden by the $\Theta$ notation are not larger than $1 + o(1)$, not smaller than $0.93 - o(1)$ for goods, and not smaller than $0.52 - o(1)$ for chores. The reason why our constant for goods is better than the one for chores is at least partly due to the fact that we spent more effort in trying to optimize the first constant.


\subsection{Related work} 

Known bounds on $\mu(n)$ for the special case of additive valuations over goods appear in Table~\ref{tab:MMS}. {(The bounds from~\cite{ezra2026} are very recent and were not known at the time when the results of the current paper were obtained.)} For large $n$, the lower bound in the first part of our Theorem~\ref{thm:main}, which applies to all monotone valuations over goods, is stronger than the previously known lower bound for this special case. (For additive valuations, we slightly improve it further in Theorem~\ref{thm:additive}.)

\begin{table}[hbt!]
        \centering
\renewcommand{\arraystretch}{1.4}%
\begin{tabular}{|l||l|l|}
 \hline
 \multicolumn{3}{|c|}{$\mu(n)$ for additive valuations over goods} \\
 \hline
  & lower bound & upper bound  \\
 \hline\hline
 2 agents   & always feasible                  & always feasible  \\
 \hline
 3 agents  & 8  \; \; \cite{FST21Negative}               & 8  \; \; \cite{FST21Negative} \\
 \hline
 4 agents         & 10  \; \; \cite{hummel2023lowerboundsmaximinshare}                        & 10 \; \; \cite{ezra2026} \\
 \hline
 $n$ agents & $n + \Omega\!\left(\frac{\log n}{\log\log n}\right)$ \; \; \cite{hummel2023lowerboundsmaximinshare} & $n + O(\log n)$ \; \; \cite{ezra2026} \\
 \hline
\end{tabular}
        \caption{
        The lower bound of~\cite{hummel2023lowerboundsmaximinshare} states that for integer $c$, if $n \ge 0.6597^c \cdot c!$, then every instance with $n$ agents and at most $n+c$ items has an MMS allocation. }
        \label{tab:MMS}
\end{table}

The MMS is a {\em share-based} fairness notion. A weaker version of the MMS, referred to as the {\em residual maximin share} (RMMS) is feasible for all monotone valuations over goods~\cite{Feige25RMMS}.

For {\em comparison-based} notions of fairness, it is known that allocations that are {\em envy free up to one good} (EF1) exist for all monotone valuations and any number of items~\cite{LiptonMMS04}, and this extends also to {\em envy free up to one less preferred good} (EFL). Moreover, allocations that are simultaneously EFL and RMMS exist for all monotone valuations over goods~\cite{Feige25RMMS}. In contrast, for every $n\ge 3$ and $m \ge n+5$, there are allocation instances with submodular valuations in which there is no allocation that is {\em envy free up to any good} (EFX)~\cite{AkramicounterEFX, mackenzie2026counterexample}. For additive valuations, there is no known allocation instance for which EFX allocations do not exist, and it is known that EFX allocations exist when $n=3$~\cite{ChaudhuryGM24EFX}. 

For some restricted classes of valuation functions, many fairness notions are feasible. See for example~\cite{BEF21}.

As MMS allocations do not always exist, there is extensive literature concerning existence of allocations that approximate the MMS in various senses, and for various classes of valuation functions. See for example~\cite{kurokawa2018fair, GhodsiHSSY22BeyondAdditive, HosseiniOrdinal}. For more information on fair division of indivisible items, see for example~\cite{amanatidis2023Survey}. 

The problems studied in this work correspond to natural extremal problems concerning {\em rainbow matchings} in hypergraphs. See more details in Section~\ref{sec:rainbow} in the appendix. 

\subsection{Notation and preliminaries}


For a function $f(x)$, we use the notation $f(x) \le \otilde(g(x))$ if there is some function $p(x)$ that is polynomial in $x$ such that $f(x) \le p(x)(1 + g(x))$ for every $x$. If the growth rate of $f(x)$ is exponential in $x$, then the term $p(x)$ is a low order multiplicative factor.

A {\em hypergraph} $G(V,E)$ has a set $V$ of vertices, and a set $E$ of hyperedges, where each hyperedge is a subset of at least two vertices. (If every hyperedge contains exactly two vertices, then the hypergraph is simply a graph.) A {\em matching} in a hypergraph is a collection of disjoint hyperedges. An {\em independent set} in a hypergraph is a set of vertices that does not contain any hyperedge.

We shall make use of the following known fact about matchings in graphs. For completeness, we sketch its proof.

\begin{lemma}
\label{lem:matching}
    Every $n$ by $n$ bipartite graph with minimum degree at least $\frac{n}{2}$ has a perfect matching.
\end{lemma}

\begin{proof}[sketch]
By Hall's condition, the only obstruction for a perfect matching is a set $S$ of vertices on one side that has fewer than $|S|$ neighbors on the other side. Under the minimum degree condition, in $0 < |S| \le \frac{n}{2}$ then even just a single vertex in $S$ has at least $|S|$ neighbors, whereas if $|S| > \frac{n}{2}$ then every vertex on the other side has at least one neighbor in $S$, and so the neighborhood of $S$ has size $n \ge |S|$. 
\end{proof}

A set function $f$ is submodular if it has the diminishing marginal value property: for every item $e$ and sets $S \subset T$, $f(S \cup \{e\}) - f(S) \ge f(T \cup \{e\}) - f(T)$.

For chores, instead of considering the non-positive valuation function $v_i$, we shall consider the non-negative cost function $|v_i| = -v_i$. Relative to $|v_i|$, in an MMS allocation agent $i$ needs to get a bundle of cost at most her MMS (which corresponds to value at least her MMS).

All logarithms are in base~2, unless stated otherwise.

\section{Proofs for the main theorem}

We prove the parts of Theorem~\ref{thm:main} in an order that introduces the ideas in the proofs in a gradual manner.

\subsection{Arbitrary valuations}

We first prove an upper bound on $\mu(n)$ that is applicable for arbitrary valuations. {We remark that a bound of $\mu(n) \le 2\log n$ can be inferred from Lemma 1.10 in~\cite{AharoniCH15}, but we are interested in a stronger upper bound.}  

\begin{theorem}
\label{thm:arbitraryNo}
    For every $n$, there is an instance with arbitrary valuations and not more than $\log n + O(1)$ items that has no MMS allocation.  
\end{theorem}

\begin{proof}
    Let $\items$ be a set of $m$ items. Designate one item $e$ as special. Consider all partitions $(S_i,T_i)$ of $\items$ into two disjoint parts, with $|S_i| \ge |T_i|$ and $e\in T_i$. There are at least {$(1 - o(1))2^{m-2}$} such partitions. With each such partition we associate an agent $i$. Hence the total number of agents $n$ is roughly $2^{m-2}$, and $m \le \log n + O(1)$.
    
 For every agent $i$ we define $v_i(S_i) = v_i(T_i) = 0$, and for every other non-empty set $S \subset \items$ (including $\items$ itself), we define $v_i(S) = -1$.

    For every agent $i$, her MMS value is $0$. There is a unique MMS partition that gives $i$ this value: $n-2$ empty parts, and the remaining two parts are $S_i$ and $T_i$.

    In this instance, there is no MMS allocation. In every allocation, there is at least one agent, say agent $i$, that receives a non-empty bundle. To get a value of $0$, agent $i$ must receive one of the sets $S_i$ or $T_i$. But then, there must be some additional agent $j$ that receives a non-empty bundle. It cannot be that $i$ receives $T_i$ and $j$ receives $T_j$, as these sets intersect at $e$. It cannot be that $i$ receives $S_i$ and $j$ receives $S_j$, as these sets also intersect (their union does not contain $e$, but the sum of their sizes is at least $m$). Thus, without loss of generality, $i$ receives $S_i$ and $j$ receives $T_j$. But as $T_j \not= T_i$, either $S_i$ and $T_j$ intersect, which is not allowed, or there still remains some unallocated item. But then, the agent $k \not\in \{i,j\}$ that receives items gets value $-1$. This is because $k$ cannot receive $S_k$ (it intersects with $S_i$), and cannot receive $T_k$ (it intersects with $T_j$).
\end{proof}

We now prove a nearly matching lower bound on $\mu(n)$.

\begin{theorem}
\label{thm:arbitraryYes}
    For some universal constant $c > 0$, for every $n$ and $m \le \log n - c$, every instance with arbitrary valuations has an MMS allocation.  
\end{theorem}

\begin{proof}
Consider an allocation instance with a set $\items$ of $m$ items and a set $\agents$ of $n > m$ agents. For every agent $i$, her MMS partition partitions $\items$ into $1 \le d_i \le m$ parts. Refer to the collection of partitions of the $n$ agents as $C$. 



{For the sake of contradiction, we assume that the allocation instance has no MMS allocation. Based on this assumption, we upper bound $n$. The assumption of having no MMS allocation is equivalent to the following assumption.}

\begin{assumption}
    \label{ass:N_d}
  For every integer $d$, there is no partition of $\items$ into $d$ non-empty parts $T_1, \ldots, T_d$ and a choice of $d$ different agents $a_1, \ldots, a_d$ in $\agents$, such that for every $1 \le i \le d$, $T_i$ is a part in the MMS partition of $a_i$.  
\end{assumption}

If Assumption~\ref{ass:N_d} fails to hold, then the instance has an MMS allocation (the remaining agents get empty bundles). 
Based on this assumption, we provide an upper bound on $n$. 

For an integer $j\ge 1$, call a set $S\subseteq \items$ \emph{$j$-popular}
if there are at least $j$ partitions in $C$ such that, in each of these
partitions, one part is $S$. Otherwise call $S$
\emph{$j$-unpopular}. We say that a partition
\[
Q=(T_1,\ldots,T_d)\in C
\]
is {\em good}, if after ordering its parts in non-increasing order of size, every
part $T_j$ is $j$-popular. Assumption~\ref{ass:N_d} implies that no partition $Q \in C$ is good. This is because for a good $Q$, visiting the indices $j$ in increasing order, we can allocate each $T_j$ to a different agent that desires it. 

Order the parts of every partition in $C$ in non-increasing order of size.
The first part is always $1$-popular, because it contains itself and hence is
witnessed by the partition in which it appears. Therefore, for every
partition in $C$, there exists some index $j\ge 2$ such that the $j$th part
is $j$-unpopular.

We now classify all partitions of $C$ into $d-1$ classes, $C_2, \ldots, C_d$. 


For any partition in $C$, let $j \ge 2$ be an arbitrary index such that its $j$th part is $j$-unpopular (such an index must exist), and place the partition in $C_j$.

Fix a partition
\[
Q=(T_1,\ldots,T_d)\in C_j.
\]

Since the parts are ordered in
non-increasing order, the first $j$ parts each have size at least
$|T_j|$. Consequently, $|T_j| \le \frac{m}{j}$. Define $s_j = \lfloor \frac{m}{j} \rfloor$.

There are at most $\sum_{i=1}^{s_j} \binom{m}{i}$ subsets of $\items$ of size at most $s_j$. {Each subset that is the $j$-th part of a partition in $C_j$ is $j$-unpopular. It appears as a part in fewer than $j$ partitions, and hence can be chosen at most $j-1$ times in the above selection.} Therefore, 
$$|C_j| \le (j-1)\sum_{i=1}^{s_j} \binom{m}{i} \le O(j \sqrt{m}) \cdot \binom{m}{s_j}.$$ 

It follows that 

\begin{equation}
\label{eq:N}
    n = \sum_{j=2}^m |C_j| \le \sum_{j=2}^m O(j \sqrt{m}) \cdot \binom{m}{s_j} \le O(2^m)
\end{equation}

(In fact, we have the slightly tighter bound of 
$N \le \sum_{j=2}^m (j-1) \cdot \sum_{i > s_{j-1}}^{s_j} \binom{m}{i}$.)

Thus, if $n \ge 2^{m + c}$ for some sufficiently large constant $c$, Assumption~\ref{ass:N_d} is violated, implying that an MMS allocation exists.
\end{proof}

\subsection{Negative examples for monotone valuations}

The upper bound on $\mu(n)$ that we prove for monotone valuations over goods applies already in the special case of submodular valuations. 

\begin{theorem}
\label{thm:submodular}
    For every $n$, there is an instance with submodular valuations and not more than $n + \log n + O(\log\log n)$ goods that has no MMS allocation. Moreover, the valuations can be chosen to be integer valued, with $v_i(\items) = O(\log n)$ for all agents. 
\end{theorem}

\begin{proof}
    For even $c$, we choose $n = \frac{1}{2}\binom{c+2}{\frac{c+2}{2}}$.  Conversely, for a given $n$, we choose the smallest even $c$ satisfying $\frac{1}{2}\binom{c+2}{\frac{c+2}{2}} \ge n$. Thus, $c \le \log n + O(\log\log n)$. We let $m = n + c$.

    It suffices to present our construction in the case that $c$ is even.
    In our instance, there is a valuation function $v$, such that for every agent $i$, valuation $v_i$ agrees with $v$ on the values of all bundles, except for two bundles (whose identity depends on $i$). 
    
    We describe $v$ first. There is a set $S$ of $c+2$ ``small items". Every item not in $S$ is ``large" and has value $c+2$. Also, $v(\items) = c+2$. For any subset $T \subset S$, we have that: 

    \begin{itemize}
        \item $v(T) = 2|T|$ if $|T| \le \frac{c}{2}$,
        \item $v(T) = c + 1$ if $|T| = \frac{c + 2}{2}$,
        \item $v(T) = c + 2$ if $|T| > \frac{c + 2}{2}$.
    \end{itemize}
      
    With every agent $i$ we associate a distinct partition of $S$ into two equal size sets. One of the two sets is denoted by $T_i$, and the other is $S \setminus T_i$.  (The value of $n$ was chosen so as not to exceed the number of such partitions.) To obtain $v_i$ from $v$, we only change the values of the sets $T_i$ and $S \setminus T_i$, setting $v_i(T_i) = v_i(S \setminus T_i) = c+2$. Observe that the resulting $v_i$ integer valued, submodular, and bounded by $c + 2 = O(\log n)$, as desired.

    For every agent $i$, her MMS value is $c+2$. There is a unique MMS partition that gives $i$ this value: $n-2$ parts each contain a single large item, and the remaining two parts are $T_i$ and $S \setminus T_i$.

    In this instance, there is no MMS allocation. In every allocation, there are at least two agents, say $i$ and $j$, that receive only items from $S$. To get a value of $c+2$, each of them must receive at least $\frac{c+2}{2}$ items from $S$, meaning that each receives exactly $\frac{c+2}{2}$ items from $S$. Thus, $i$ should receive one of the sets $T_i$ or $S \setminus T_i$, and $j$ should receive one of the sets $T_j$ or $S \setminus T_j$. However, as no set desired by $i$ is a complement of a set desired by $j$, at least one of the two agents does not get a desired set.
\end{proof}

The proof of Theorem~\ref{thm:submodular}, with straightforward modifications, extends also to chores. The cost function $|v_i|$ over chores is derived from the valuation $v_i$ over goods, with the following differences. Large items have cost $c+1$. For sets $T$ of size $\frac{c+2}{2}$, $|v_i|(T) = c+2$ if $v_i(T) = c+1$, and $|v_i|(T) = c+1$ if $v_i(T) = c+2$. The MMS cost is $c+1$, but in every allocation, some agent gets a set of cost $c+2$.

\subsection{Existence of MMS allocations over goods}

In this section, we prove our lower bound on $\mu(n)$ for monotone valuations over goods. We use the following key definition (the terminology $n_c$ is taken from~\cite{hummel2023lowerboundsmaximinshare}).

\begin{definition}
    \label{def:nc}
    For every integer {$c \ge 1$}, define $n_c$ to be the minimum value of $n$ for which the following holds. Consider any integer $1 \le d \le c$, any allocation instance with a set $\items$ of $c + d$ goods and a set $\agents$ of $n \ge n_c$ agents, where each agent $i$ has an arbitrary monotone valuation function $v_i$, for which the $MMS(\items, v_i, d)$ partition contains only parts of size at least~2. Then, there is a partition of $\items$ into $d$ parts $B_1, \ldots, B_{d}$, and a choice of $d$ distinct agents, $\agents' = \{a_1, \ldots, a_{d}\}$, such that  $v_{a_i}(B_i) \ge MMS(\items, v_{a_i}, d)$, for every $1 \le i \le d$. In other words, there is a subset $\agents' \subset \agents$ of $d$ agents, such that the allocation instance with items $\items$ and agents $\agents'$ has an MMS allocation. 
\end{definition}

Our lower bound on $\mu(n)$ follows easily from the following two lemmas.

\begin{lemma}
    \label{lem:ncApplication}
    For every integer $c \ge 1$ and $n_c$ as in Definition~\ref{def:nc}, for every $n \ge \max[20c, n_c + 4c]$, every instance with monotone valuations and $n + c$ goods has an MMS allocation.
\end{lemma}

A lemma similar to Lemma~\ref{lem:ncApplication}  was used in~\cite{hummel2023lowerboundsmaximinshare} for additive valuations (and with a slightly smaller value of $n$). 
For additive valuations, its proof is a straightforward consequence of the fact~\cite{bouveret2016characterizing} that for MMS allocations with additive valuations, it suffices to consider instances with {\em same order preferences} over individual items. 

Our second lemma improves upon the bound of $n_c \le 2^{O(c \log c)}$ given in~\cite{hummel2023lowerboundsmaximinshare} in two respects: our upper bound is asymptotically smaller (has only exponential dependence on $c$), and it applies to all monotone valuations and not just to additive valuations.

\begin{lemma}
    \label{lem:ncUpperBound}
    The value of $n_c$ in definition~\ref{def:nc} satisfies $n_c \le (2.0922 + o(1))^c$.
\end{lemma}

With the above two lemmas we prove the lower bound on $\mu(n)$ for monotone valuations over goods. claimed in Theorem~\ref{thm:main}.

\begin{proof}
For a given $n$, let $c$ be largest so that the inequality  $n \ge \min[20c, (2.0922 + o(1))^c + 4c)]$ holds, with the $o(1)$ term being that of Lemma~\ref{lem:ncUpperBound}. The combination of Lemmas~\ref{lem:ncUpperBound} and~\ref{lem:ncApplication} implies that every instance with $n+c$ goods has an MMS allocation. As $n$ grows, the value of $c$ approaches $\frac{\log n}{\log 2.0922} > 0.93 \log n$.  
\end{proof}

We now prove Lemma~\ref{lem:ncApplication}.

\begin{proof}
    Denote the number of items by $n+c$. We assume that $n \ge 20c$.
    
    We say that an item $e$ is {\em large} for agent $i$ if $v_i(e) \ge MMS(\items, v_i, n)$. For every $i$, there are at least $n-c$ items that are large for $i$, because
    the MMS partition for $v_i$ must contain at least $n - c$ parts that each contain a single item.

    Consider the following bipartite graph $G(\agents, \items; E)$. There are $n$ vertices on one side, each representing an agent, and $n+c$ vertices on the other side, each representing an item. For every agent $i \in \agents$ and item $e \in \items$, the graph $G$ contains the edge $(i,e)$ if and only if $e$ is large for $i$. Hence, every vertex in $\agents$ has degree at least $n -c$, and the total number of (bipartite) edges missing from $G$ is at most $2nc$. 

    If $G$ has a matching that matches all vertices of $\agents$, then this matching induces an MMS allocation. Hence, we may assume that there is no such matching.

    Remove from $\items$ any vertex $e$ that has degree smaller than $\frac{n + t}{2}$, for $t = 4c$. 
    Let $K$ denote the set of items removed, let $k$ denote their number. The upper bound on the number of edges missing from $G$ implies that $k \le \frac{4nc}{n - t}$. For our choice of parameters (that is, $n \ge 20c$ and $t =4c$), we have that $k \le 5c$. Let $G'$ denote the bipartite graph that remains, and let $\items' = \items \setminus K$ denote the set of vertices that remain in the $\items$ side. We have that $|\items'| = n+c-k$.

    We may assume $k > c$. Suppose otherwise. Then in $G'$, keep removing vertices from $\items$ until exactly $n$ vertices remain in $\items$. At that point, we have a bipartite graph with $n$ vertices on each side. Every vertex in the $\agents$ side has degree at least $n - 2c \ge \frac{n}{2}$ (as $n \ge 20c$), and every vertex in the $\items$ side has degree at least $\frac{n+t}{2}$. By Lemma~\ref{lem:matching}, this graph has a perfect matching, contradicting our assumption that not all of $\agents$ can be matched in $G$.


    \begin{lemma}
    \label{lem:residual}
       In $G'$, every set of $n+c-k$ vertices of $\agents$ has a perfect matching with $\items'$.
    \end{lemma}

    \begin{proof}
    Removing $k - c$ vertices from $\agents$, the degree of each vertex in $\items$ is at least $\frac{n+t}{2} - k + c$. The degree of each vertex that remains in $\agents$ is at least $n - c - k$. Consider the following two inequalities.
    \begin{itemize}
        \item $\frac{n+t}{2} - k + c \ge \frac{n+c-k}{2}$, or equivalently, $t + c \ge k$.
        \item $n - c - k \ge \frac{n+c-k}{2}$, or equivalently, $n \ge k + 3c$.
    \end{itemize}
    Both inequalities hold, because $t = 4c$, $k \le 5c$ and $n \ge 20c$. Lemma~\ref{lem:matching} (with $n + c - k$ instead of $n$) implies the existence of the desired perfect matching.
    \end{proof}

    Equipped with Lemma~\ref{lem:residual}, we now prove Lemma~\ref{lem:ncApplication}. Recall that $k > c$. 
    For every agent $a \in \agents$, $MMS(K, v_a, k-c) \ge MMS(\items, v_a, n)$ {(because at most $|\items'| = n+c-k$ parts of agent $a$'s $n$-partition of $\items$ contain items from $\items'$, and so at least $k-c$ parts are fully contained in $K$).}

    Recall that $k \le 5c$. If there are agents for whom some items in $K$ are large (this will necessarily happen if $k > 2c$, but may also happen if $k \le 2c$),  we use the following preliminary procedure.  
     As long as possible, give an agent a large item from $K$, removing the agent and the item. If this process ends with $c$ or fewer items, then the argument preceding Lemma~\ref{lem:residual} completes the proof. Hence, we may assume that the set $K'$ of remaining items satisfies $|K'| = c + d$, for some $d$ in the range $[1, c]$. At this point, the $MMS(K', v_a, d)$ partition of every remaining agent $a$ contains only parts of size at least~2. The number of remaining agents is at least $|\agents| - 4c + d \ge n_c$. 
    
    We have reached a stage where Definition~\ref{def:nc} implies that there is a subset $\agents' \subset \agents$ of $d$ agents, such that the allocation instance with items $K'$ and agents $\agents'$ has an MMS allocation. Lemma~\ref{lem:residual} implies that the $n + c - k$ remaining  agents can each receive a large item. Thus, we have an MMS allocation for $\agents$ with the items $\items$. 
\end{proof}

We wish to prove that $n_c \le (b + o(1))^c$ for some constant $b$ that is as close as possible to~2 (or better still, $b=2$).  Lemma~\ref{lem:ncUpperBound} claims that we can take $b = 2.0922$. Conceptually, one may think of its proof as being composed of two parts. One part is to get a crude upper bound on $b$, and the other part is to optimize this upper bound.

The first part can be obtained using a straightforward adaptation of the proof of Theorem~\ref{thm:arbitraryYes}. As $n_c$ relates to a setting in which the number of items is at most $2c$, the proof of Theorem~\ref{thm:arbitraryYes} implies that we may take $n_c = O(2^{2c})$, implying that we can take $b = 4$.

The main work in the proof of Lemma~\ref{lem:ncUpperBound} is to reduce $b$ to $2.0922$. This is significantly more complicated, and involves the introduction of new ideas. In this section we only outline the definitions and intermediary steps that we use. Full proofs are presented in Section~\ref{sec:goods} in the appendix.

To prove Lemma~\ref{lem:ncUpperBound}, we introduce a parameter $\alpha(q)$ that we associate with hypergraphs. The relation between $\alpha(q)$ and $n_c$ will be made clear in  Lemma~\ref{lem:hypergraphReduction}. 

\begin{definition}
    \label{def:alpha}
    The parameter $\alpha(q)$ is defined in several steps.

    \begin{enumerate}
        \item A hypergraph $G(V,E)$ will be called $(q,p)$ if it has $q + p$ vertices, and its maximum matching  contains either $p$ or $p+1$ hyperedges. The definition implies that $0 \le p \le q$.

\item In a $(q,p)$ hypergraph, two disjoint sets of vertices, $S$ and $T$, are called a {\em free pair} if each of them is an independent set in $G$, and the size of the maximum matching in the sub-hypergraph induced on $V \setminus (S \cup T)$ is $p$. 

\item A collection of free pairs $(S_1,T_1), (S_2,T_2), \ldots$ is called {\em an anti-chain} if no set that is a member of a pair in the collection (no set $S_i$ or $T_i$) contains (or is equal to) another such set (a set $S_j$ or $T_j$). 

\item $\alpha(G)$ denotes the maximum number of free pairs in an anti-chain of free pairs in $G$. We define $\alpha(q,p) = \max_{G}[\alpha(G)]$, where $G$ ranges over all $(q,p)$ hypergraphs. Finally, we define $\alpha(q) = \max_{p}[\alpha(q,p)]$. 
\end{enumerate}
\end{definition}

The following lemma connects between $\alpha(q)$ of Definition~\ref{def:alpha} and $n_c$ of Definition~\ref{def:nc}.

\begin{restatable}{lemma}{hypergraphReduction}
    \label{lem:hypergraphReduction}
    $n_c \le \alpha(q) + (1.91 + o(1))^c$, for $q = c+2$. 
\end{restatable}

We prove the following upper bound on $\alpha(q)$.

\begin{restatable}{lemma}{alphaBound}
\label{lem:alpha}
$\alpha(q)\le (2.0922+o(1))^q$.
\end{restatable}

The combination of Lemmas~\ref{lem:hypergraphReduction} and~\ref{lem:alpha} implies Lemma~\ref{lem:ncUpperBound}.


The worst case instances for our proof in Lemma~\ref{lem:alpha} 
are when the hypergraph $G(V,E)$ is in fact a graph. In this natural special case, Theorem~\ref{thm:graph} provides the optimal upper bound of $\alpha(q) \le (2+o(1))^q$. This indicates that the bounds in Lemma~\ref{lem:alpha} can be  improved.

\begin{restatable}{theorem}{graph}
\label{thm:graph}
If \(G\) ranges over all graphs (rather than hypergraphs), then
\[
\alpha(q)\le (2+o(1))^q .
\]
\end{restatable}

\subsection{Existence of MMS allocations for chores}

We prove here a lower bound on $\mu(n)$ for monotone valuations over chores. Its proof uses a decomposition into two steps similar to the case of goods. The first step is the analog of Lemma~\ref{lem:ncApplication}, that shows that we can focus on subsets that contain $O(c) = O(\log n)$ items. However, for chores, the proof of this fact is different from the proof in the case of goods, and makes up the first part of the proof of Theorem~\ref{thm:chores}.  
The second step involves the analysis over small sets of items.  Here, our proof of Theorem~\ref{thm:chores} uses only a rather crude bound. The bound can be somewhat improved by using optimizations that we applied for goods,
but this would come at the cost of complications in the proofs.  

\begin{theorem}
    \label{thm:chores}
     Every instance with monotone valuations over chores and at most $n + (0.52 - o(1))\log_2 n$ items has an MMS allocation. 
\end{theorem}

\begin{proof}
    Consider an allocation instance with monotone valuations over chores, with a set $\agents_1$ of $n_1$ agents and a set $\items_1$ of $n_1+c_1$ items, with $n_1 > c_1$.

    In our proof, the sets $\agents_1$ and $\items_1$ may change, where after step $t$ in the proof they become $\agents_{t+1}$ and $\items_{t+1}$. We shall have $n_t = n_1 - t + 1$ and $c_t \le c_1 - t + 1$. In addition we will keep track of a subset $\agents'_t \subset \agents_t$ of {\em active agents}. Initially, $\agents'_1 = \agents_1$. At every step $t$, the number of parts in an MMS partition of active agents is $|\agents_{t}|$  (not to be confused with $|\agents'_{t}|$). We also fix some $0 < \gamma < 1$, to be optimized later. We assume that initially $(1-\gamma)^{c_1} n_1 \ge 1$

    Fix an arbitrary step $t$.
    
    Consider an arbitrary active agent $i \in \agents'_{t}$. As $n_t > c_t$, some parts in $i$'s MMS partition are {\em singletons} (contains one item), and others have larger size. Call an item {\em $i$-large} if it is one of the singletons of the MMS partition, and {\em $i$-small} otherwise. At least $n_t-c_t$ items are $i$-large, and at most $2c_t$ items are $i$-small.

Call an item {\em large} if it is $i$-large for every $i \in \agents'_t$.  Call an item {\em often-small} if it is $i$-small for at least $\lceil \gamma |\agents'_t| \rceil$ agents in $\agents'_t$. The remaining items (those that are small for between~1 and  $\gamma |\agents'_t|$ agents in $\agents'_t$) are called {\em sometimes-small}. 

    Suppose first that there is a sometimes-small chore $e$. Let $\agents^e_t$ denote the set of agents in $\agents'_t$ for which $e$ is large, and note that $|\agents^e_t| > (1 - \gamma)|\agents'|$. Let $i\in \agents'_t$ be an agent for which $e$ is $i$-small, and let $S$ be the bundle in agents $i$ MMS partition that contains $e$. As $e$ is $i$-small, we have that $|S| \ge 2$. Give $S$ to $i$. 
    
    Consider the allocation instance that remains for step $t+1$    
    We have that $\agents_{t+1} = \agents_t \setminus \{i\}$. Thus $n_{t+1} = n_t - 1$. We also have that $\items_{t+1} = \items_t \setminus S$. This implies that $c_{t+1} = c_t - |S| + 1 \le c_t - 1$. As for active agents, we define $\agents'_{t+1} = \agents^e_t$, and consequently, $|\agents'_{t+1}| > (1 - \gamma)|\agents'_t|$. Importantly, for every agent $j \in \agents'_{t+1}$, $MMS(\items_{t+1}, |v_j|, n_{t+1}) \le MMS(\items_t, |v_j|, n_t)$. This is because the following is a feasible option for an $MMS(\items_{t+1}, |v_j|, n_{t+1})$ partition: 
    take the  $MMS(\items_t, |v_j|, n_t)$ partition and remove from it all items of $S$. As $e \in S$ and $e$ was a singleton in the  $MMS(\items_t, |v_j|, n_t)$ partition, we can drop the singleton part $\{e\}$ (thus the partition now has $n_{t+1} = n_t - 1$ parts), and no other part gained any new chore (so its cost could not increase). 

    The above procedure ends at a step $t$ in which no item is sometimes-small. The following claim summarizes the situation at that step.

    \begin{claim}
    \label{cl:chores}
        For every $t$, $|\agents_t| = n_t = n_1-t+1$, $c_t \le c_1 - t + 1$, $|\items_t| = n_t + c_t$, and there is a set $\agents'_t$ of active agents, with $|\agents'_t| \ge (1 - \gamma)^{t-1}n_1$. The above facts imply that $(1-\gamma)^{c_t} |\agents'_t| \ge 1$. For every agent $j\in \agents'_t$, $MMS(\items_{t}, v_j, n_t) \le MMS(\items_1, v_j, n_1)$.
    \end{claim}

    At step $t$, every item is either large or often-small. We may assume that some items are often-small. (Otherwise, all items are large, implying that {$n_t \ge |\items_t|$}. Giving each agent of $\agents_t$ at most one item, every agent gets a bundle of cost at most her MMS.) Let $S_t$ denote the set of often-small items.

    We first provide an upper bound on $|S_t|$. Every agent in $\agents'_t$ contributes at most $2c_t$ items to $S_t$. Every item in $S_t$ needs to be contributed by at least $\gamma |\agents'_t|$ agents. Hence, $|S_t| \le \frac{2c_t |\agents'_t|}{\gamma |\agents'_t|} \le \frac{2c_t}{\gamma}$.

    Let $d = |S_t| - c_t$. The $MMS(\items_t, v_i, n_t)$ partition of each agent $i\in \agents'_t$ partitions $\items_t \setminus S$ into $|\items_t \setminus S| = n_t - d$ singletons, and hence it partitions $S_t$ into $d$ parts.

    Assume that there is partition of $S_t$ into $d$ parts $T_1, \ldots, T_d$ and a choice of $d$ different agents $a_1, \ldots, a_d$ in $\agents'_t$, such that for every $1 \le i \le d$, $T_i$ is a part in the MMS partition of $a_i$.  Then the instance has an MMS allocation (the remaining agents each gets a single item). 

    So now we can proceed as in the proof of Theorem~\ref{thm:arbitraryYes}, so as to show that if $|\agents'_t|$ is sufficiently large, the assumption holds. The second largest part in a $d$-partition of $S_t$ has size at most $\frac{|S_t| - (d-2)}{2} \le \frac{c_t}{2} + 1 \simeq \frac{c_t}{2}$. Similar to the computations leading to Equation~\ref{eq:N} in the proof of Theorem~\ref{thm:arbitraryYes}, we conclude that is suffices that $|\agents'_t| \ge \otilde(\binom{|S_t|}{c_t/2}) \le \otilde(\binom{2c_t/\gamma}{c_t/2}) \simeq \otilde(2^{2c_t H(\gamma/4)/\gamma})$, where $H$ is the binary entropy function.

    Recall that we also have that $|\agents'_t| \ge 1/(1 - \gamma)^{c_t}$. Thus, up to low order terms, it suffices that we choose $\gamma$ such that $\frac{1}{1 - \gamma} \ge 2^{2 H(\gamma/4)/\gamma}$. This gives $\gamma \simeq 0.73$, and $|\agents_t| \simeq \frac{1}{0.27^{c_t}} \simeq 2^{1.9 c_t}$.

    Consequently, it suffices that $n_1 \ge \otilde(2^{1.9 c_1})$, which holds for $c_1 \le (0.52 - o(1)) \log n_1$.
\end{proof}

\section{Discussion}

It is an interesting open question whether for monotone valuations over goods, $\mu(n) = n + \log n + o(\log n)$. This relates to tightening the bounds on $\alpha(q)$ from Definition~\ref{def:alpha}.
Ideas from the proof of Theorem~\ref{thm:graph} can be used in order to improve over the bounds of Lemma~\ref{lem:alpha} for general hypergraphs.
However, handling hypergraphs is considerably more complicated than handling graphs. In particular, maximum matchings do not enjoy structural results similar to the  Gallai–Edmonds decomposition for graphs. We were only able to obtain minor improvements over the bounds of Lemma~\ref{lem:alpha}, and leave it as an open question whether $\alpha(q) \le (2 + o(1))^q$.

Limiting the number of goods to $ n + (1 + o(1))\log n$ is necessary if one wants to guarantee existence of MMS allocations, even in the special case of submodular valuations (Theorem~\ref{thm:submodular}).
However, in the important special case of additive valuations, $n + (1 + o(1))\log n$ is not a barrier. 
In Section~\ref{sec:additive} in the appendix we prove the following theorem.

\begin{theorem}
    \label{thm:additive}
    For additive valuations, $\mu(n) \ge n + (1.088 - o(1))\log n$.
\end{theorem}  

At the time when Theorem~\ref{thm:additive} was proved, the known upper bound for additive valuations was $\mu(n) \le 3n + O(1)$.
Very recently, it was shown in~\cite{ezra2026} that $\mu(n) \le n + O(\log n)$ for additive valuations. 

\subsection*{Acknowledgements}

This research was supported in part by the Israel Science Foundation (grant No. 1122/22).
The process of improving the constants in terms of the form $\Omega(\log n)$ was assisted by a version of ChatGPT, based on OpenAI’s GPT-5.3–class model. 

\bibliographystyle{alpha}


\newcommand{\etalchar}[1]{$^{#1}$}

\begin{appendix}
    \section{An extremal problem on hypergraphs}
    \label{sec:goods}

    In this section, we provide the full proof of Lemma~\ref{lem:ncUpperBound}. The proof is based on analysing an extremal problem on hypergraphs, presented in Definition~\ref{def:alpha}.

    We restate and prove the lemma that connects between $\alpha(q)$ of Definition~\ref{def:alpha} and $n_c$ of Definition~\ref{def:nc}.

\hypergraphReduction*

\begin{proof}
Consider the setting of Definition~\ref{def:nc}. $\items$ is a set of $c + d$ items, there are $N$ agents, and each agent has a partition of $\items$ into $d$ parts, each with at least two items. We refer to the collection of $d$-partitions of the $N$ agents as $C$. 

\begin{assumption}
    \label{ass:N}
    Under the valuations of the agents, the conclusion of Definition~\ref{def:nc} fails to hold (there is no subset of $d$ agents for which $\items$ gives an MMS allocation).
\end{assumption}

Based on this assumption, we provide an upper bound on $N$. Adding one to this upper bound, we get an upper bound on $n_c$.

For an integer $j\ge 1$, call a set $S\subseteq \items$ \emph{$j$-good}
if there are at least $j$ partitions in $C$ such that, in each of these
partitions, at least one part is contained in $S$. Otherwise call $S$
\emph{$j$-bad}. We say that a partition
\[
Q=(T_1,\ldots,T_d)\in C
\]
is {\em good}, if after ordering its parts in non-increasing order of size, every
part $T_j$ is $j$-good. Assumption~\ref{ass:N} implies that no partition $Q \in C$ is good. This is because for a good $Q$, visiting the indices $j$ in increasing order, we can allocate each $T_j$ to a different agent that desires it, satisfying the conclusion of Definition~\ref{def:nc}. 

Order the parts of every partition in $C$ in non-increasing order of size.
The first part is always $1$-good, because it contains itself and hence is
witnessed by the partition in which it appears. Therefore, for every
partition in $C$, there exists some index $j\ge 2$ such that the $j$th part
is $j$-bad.

Classify all partitions of $C$ into $d-1$ classes, $C_2, \ldots, C_d$. 

Class $C_2$ contains all the partitions for which all parts but the first two are $d$-good. Observe that every partition in $C_2$, the second part is necessarily 2-bad, and the last $d-2$ parts are $d$ good.

For any partition in $C \setminus C_2$, let $j \ge 3$ be an arbitrary index such that its $j$th part is $d$-bad (such an index must exist), and place the partition in $C_j$.

Fix a partition
\[
Q=(T_1,\ldots,T_d)\in C_j.
\]

We first bound the size of $T_j$. Since the parts are ordered in
non-increasing order, the first $j$ parts each have size at least
$|T_j|$, and the remaining $d-j$ parts each have size at least $2$.
Thus
\[
j|T_j|+2(d-j)\le c+d.
\]
Consequently
\[
|T_j|\le \frac{c+d-2(d-j)}{j}
=
\frac{c-d}{j}+2,
\]
and therefore
\[
|T_j|\le s_j:=\left\lfloor\frac{c-d}{j}\right\rfloor+2.
\]

Observe that $N = \sum_{j=2}^d |C_j|$. For $j \ge 3$, the following crude upper bound on $|C_j|$ will suffice for our purpose. 

There are at most $\sum_{i=2}^{s_j} \binom{c+d}{i}$ subsets of $\items$ of size at most $s_j$. Since each subset is $d$-bad, it appears as a part in fewer than $d$ partitions, and hence can be chosen at most $d-1$ times in the above selection. Therefore, 
$$|C_j| \le (d-1)\sum_{i=2}^{s_j} \binom{c+d}{i} \le (d-1)s_j \binom{c+d}{s_j} \le \frac{(d-1)(c-d)}{j}\binom{c+d}{\frac{c-d}{j}}.$$ 

For all $j \ge 3$, $\binom{c+d}{\frac{c-d}{j}} \le (1.91 + o(1))^{c}$. (This was computed numerically. The worst choice of $j$ is $j = 3$, and then the worst choice of $d$ is around $0.032 c$.) Thus, we have that:

$$N = |C_2| + \sum_{j=3}^d |C_j| < |C_2| + c^3(1.91 + o(1))^{c} \le |C_2| + (1.91 + o(1))^{c}$$

To prove the lemma, it remains to relate $|C_2|$ to $\alpha(q)$, for $q= c+2$. We explain how to construct the $(q,p)$ hypergraph required by Definition~\ref{def:alpha} and the anti-chain of free pairs. 

Set $q = c+2$ and $p = d - 2$. The vertices of the hypergraph $G(V,E)$ are the items of $\items$. Hence, the number of vertices is $c+d = q + p$, as desired. The hyperedges of $G(V,E)$ are derived from the partitions of $C_2$. Each partition provides $d-2 = p$ hyperedges that form a matching in $G$. The same hyperedge may be provided by several partitions. Specifically, from each $d$-partition  $Q = (T_1, \ldots, T_d) \in C_2$, the items in each one of the parts $T_3, \ldots, T_d$ form a hyperedge. Thus, $G$ necessarily has a matching of size $d-2 = p$. The size of its maximum matching is at most $p+1$, because if it has a matching of size $p+2 = d$, Assumption~\ref{ass:N} is violated. (As each of the $d$ hyperedges of the matching is desired by at least $d$ agents, there is a partition of $\items$ into $d$ disjoint bundles (each matching hyperedge in a different bundle), and an allocation of the parts of this partition to $d$ agents, so that each of these agents desires that part that she receives.)

For the anti-chain of free pairs, we take the first two parts of every partition in $C_2$. That is, from every  $Q=(T_1,\ldots,T_d)\in C_2$, we take the pair $(T_1, T_2)$. The sub-hypergraph induced on $V \setminus (T_1 \cup T_2)$ contains a matching of size $p$, namely, the hyperedges $T_3, \ldots, T_d$. It does not contain a matching of size $p+1$, as then Assumption~\ref{ass:N} is violated. Moreover, $T_2$ is an independent set in $G$ by definition of $C_2$, and so is $T_1$ (otherwise, Assumption~\ref{ass:N} is violated). Hence, indeed $(T_1,T_2)$ is a free pair in $G$.

Finally, we verify that the collection of free pairs is an anti-chain. Suppose for contradiction that there is some $Q'=(T'_1,\ldots,T'_d)\in C_2$ such that the pairs derived from $Q$ and $Q'$ violate the anti-chain property. For concreteness, suppose that $T'_2 \subset T_1$ (other cases are treated similarly). Then we can give $T_1$ to the agent having partition $Q'$, $T_2$ to the agent having partition $Q$, and each of $T_3, \ldots, T_d$ to distinct other agents (as each of these parts is desired by at least $d$ agents), violating Assumption~\ref{ass:N}.

Thus, $D$ has an anti-chain of size $|C_2|$, implying by Definition~\ref{def:alpha} that $|C_2| \le \alpha(q)$. This implies that $N \le \alpha(q) + (1.91 + o(1))^{c}$, proving the lemma.
\end{proof}

We first present a crude upper bound on $\alpha(q)$, similar to the crude upper bound on $|C_2|$ that is given in the proof of Lemma~\ref{lem:hypergraphReduction}

\begin{claim}
    \label{cl:alpha}
    $\alpha(q) \le (2.148 + o(1))^q$.
\end{claim}

\begin{proof}
Let $G(V,E)$ be an $(q,p)$ hypergraph and let $C$ be an anti-chain of free pairs (as in Definition~\ref{def:alpha}).
For every free pair $(S_i,T_i) \in C$, we assume without loss of generality that $|T_i| \le |S_i|$. Definition~\ref{def:alpha} implies that $|T_i| \le \frac{q+p - 2p}{2} = \frac{q-p}{2}$. 
{As the $T_i$ are all distinct,  we have that $|C| \le \sum_{i=1}^{(q-p)/2} \binom{q+p}{i} < q\binom{q+p}{\frac{q-p}{2}}$.}
{(In fact, the $T_i$ form an anti-chain of sets, which implies the stronger bound of $|C| \le \binom{q+p}{\frac{q-p}{2}}$ (by a certain inequality known as the  Lubell–Yamamoto–Meshalkin inequality). However, this saving of a multiplicative term of $q$ has only negligible effect on the computations that follow -- it is absorbed in the final $o(1)$ term.)} 

Write $p = \gamma q$ with $\gamma\in[0,1]$. Then
\[
\binom{q+p}{\frac{q-p}{2}}
= \binom{(1+\gamma)q}{\frac{1-\gamma}{2}q}.
\]
Using the standard entropy estimate,
\[
\binom{N}{\alpha N}
= \exp\!\big(N H(\alpha) + o(N)\big),
\]
we obtain
\[
\binom{(1+\gamma)q}{\frac{1-\gamma}{2}q}
= \exp\!\left((1+\gamma)q \cdot H\!\left(\frac{1-\gamma}{2(1+\gamma)}\right) + o(q)\right).
\]

This expression is maximized when $\gamma\simeq 0.223$, giving
\[
\max_{\gamma} (1+\gamma) H\!\left(\frac{1-\gamma}{2(1+\gamma)}\right)
< \log(2.148),
\]
and hence
\[
\alpha(q) \le (2.148 + o(1))^q.
\]
\end{proof}

We restate and prove our improved upper bound on $\alpha(q)$. 

\alphaBound*

\begin{proof}
Fix an $(q,p)$ hypergraph $G=(V,E)$, and let
\[
|V|=q+p.
\]
Let
\[
M=\{e_1,\ldots,e_p\}
\]
be a matching of size $p$ in $G$ whose union of hyperedges contains the smallest number of vertices. Write
\[
r_i=|e_i|,\qquad R=\sum_{i=1}^p r_i.
\]
Thus \(r_i\ge 2\), and \(R\) is the number of vertices covered by the matching \(M\).

Let
\[
\{(S_1,T_1),\ldots,(S_N,T_N)\}
\]
be an independent collection of free pairs, with
$|S_i|\ge |T_i|$
for every \(i\).

Since \((S_i,T_i)\) is a free pair, the hypergraph induced on
\[
V\setminus (S_i\cup T_i)
\]
contains a matching of size \(p\). In particular, it contains at least \(R\) vertices. Hence
\[
|S_i\cup T_i|\le q+p-R.
\]
Since \(|S_i|\ge |T_i|\), we get
\[
|T_i|\le \frac{q+p-R}{2}.
\]

Recall from Definition~\ref{def:alpha} that no \(T_i\) can contain a whole edge of the fixed matching \(M\). 

Therefore each \(T_i\) is a subset of \(V\) satisfying:

\[
|T_i|\le \frac{q+p-R}{2},
\]

and

\[
e_j\nsubseteq T_i
\qquad\text{for every }j=1,\ldots,p.
\]

The sets \(T_1,\ldots,T_N\) form an antichain. Up to a multiplicative factor of $\frac{q+p-R}{2} < \frac{q}{2}$, we may restrict attention to those \(T_i\)'s of one fixed size \(t\), where

\[
t\le \frac{q+p-R}{2}.
\]

For a fixed size \(t\), the number of possible choices for \(T_i\) is at most the coefficient of $x^t$ in the following generating function:

\[
(1+x)^{q+p-R}
\prod_{j=1}^p\left((1+x)^{r_j}-x^{r_j}\right).
\]

The factor
\[
(1+x)^{q+p-R}
\]
corresponds to the vertices not covered by \(M\). For a matching edge \(e_j\) of size \(r_j\), the set \(T_i\) may contain any proper subset of \(e_j\), but not all of \(e_j\). Hence the local generating function is

\[
\sum_{a=0}^{r_j-1}\binom{r_j}{a}x^a
=
(1+x)^{r_j}-x^{r_j}.
\]

Thus, for every \(x>0\), the coefficient of $x^t$ is at most

\[
\frac{
(1+x)^{q+p-R}
\prod_{j=1}^p\left((1+x)^{r_j}-x^{r_j}\right)
}{x^t}.
\]

Since the optimum occurs for \(x<1\), and \(t\le (q+p-R)/2\), this is at most

\begin{equation}
\label{eq:coefficient}
\frac{
(1+x)^{q+p-R}
\prod_{j=1}^p\left((1+x)^{r_j}-x^{r_j}\right)
}{x^{(q+p-R)/2}}.
\end{equation}

We intend to eventually choose $x < 0.6$. 

\begin{claim}
    For $x < 0.6$, the expression (\ref{eq:coefficient}) is maximized when $r_j = 2$ for all $j$.
\end{claim}

\begin{proof}
   By definition, $r_j \ge 2$, as the $r_j$ are sizes of hyperedges. 
   
   Suppose for the sake of contradiction that the maximum is attained when some $r_j \ge 3$. We contradict this by showing that decreasing $r_j$ to $r_j -1$ increases the value of the expression. 

    Decreasing $r_j$ to $r_j - 1$ also decreases $R$ by~1. The overall change in the value of (\ref{eq:coefficient}) is by a multiplicative factor of:

$$\frac{(1+x)((1 + x)^{r_j-1} - x^{r_j - 1})}{\sqrt{x}((1+x)^{r_j}-x^{r_j})} = \frac{(1 + x)^{r_j} - x^{r_j} - x^{r_j-1}}{\sqrt{x}((1+x)^{r_j}-x^{r_j})}$$
    
   For $0 < x < 1$ the term  $(1 + x)^{r_j} - x^{r_j}$ is increasing with $r_j$, whereas the term $x^{r_j-1}$ is decreasing with $r_j$. Hence the change is smallest when $r_j$ is smallest, namely, $r_j = 3$. Thus the change is by a multiplicative factor of at least $\frac{(1 + x)^3 - x^3 - x^2}{\sqrt{x}((1+x)^3-x^3)}$. As $(1 + x)^3 - x^3 > 1 + 3x^2$, this last expression is at least $\frac{1 + 2x^2}{\sqrt{x}(1+3x^2)}$. This expression decreases with $x$, and at $x = 0.6$ its value is roughly $1.07 > 1$.
\end{proof}

Replacing every $r_j$ by~2 and $R$ by $2p$ in expression (\ref{eq:coefficient}), we get

\[
\frac{
(1+x)^{q-p}
\prod_{j=1}^p\left((1+x)^2-x^2\right)
}{x^{(q-p)/2}}.
\]

Write $p = \gamma q$, with $0\le \gamma \le 1$. 

The remaining optimization is therefore

\[
\max_{0\le \gamma\le 1}
\inf_{0 < x < 0.6}
\frac{(1+x)^{1-\gamma}(1+2x)^\gamma}
{x^{(1-\gamma)/2}}.
\]

Optimizing this expression numerically gives:

\[
\gamma \approx 0.21947,
\qquad
x\approx 0.546,
\]

and the resulting value is slightly less than 2.0922.

Therefore

\[
\alpha(q)\le (2.0922+o(1))^q.
\]
\end{proof}

The combination of Lemmas~\ref{lem:hypergraphReduction} and~\ref{lem:alpha} implies Lemma~\ref{lem:ncUpperBound}.


The worst case instances for our analysis that provided upper bounds on $\alpha(q)$ (both for Claim~\ref{cl:alpha} and for Lemma~\ref{lem:alpha}) are when the hypergraph $G(V,E)$ is in fact a graph. In this natural special case, Theorem~\ref{thm:graph}, restated below, provides an upper bound of $\alpha(q) \le (2+o(1))^q$. This is tight up to low order terms. 

\graph*

\begin{proof}
Fix an \((q,p)\) graph \(G=(V,E)\), and let
\[
\{(S_1,T_1),\ldots,(S_N,T_N)\}
\]
be an antichain free pairs, with $|S_i|\ge |T_i|$ for every \(i\).
As explained in previous proofs,
\[
|T_i|\le \frac{q-p}{2}.
\]

By Definition~\ref{def:alpha}, the maximum matching in $G$ has size either $p$ or $p+1$. Assume first that its size is $p$.

Consider the {\em Gallai--Edmonds decomposition} of $G(V,E)$. It decomposes $V$ into a set $S$ of {\em separator} vertices, a collection $C_1,\ldots,C_k$ of {\em odd components} (with $c > |S|$), and a set $R$ of {\em remaining} vertices. Every odd component $C_i$ has an odd number of vertices, that we denote by $s_i$, and may have only internal edges and edges connecting it with $S$. In every maximum matching, all vertices in $R$ are matched (to each other, $|R|$ is even), all vertices in $S$ are matched (each to a vertex in a different odd component), and within each odd component $C_i$, there is an internal matching matching $s_i - 1$ vertices. We remark that $S$ and/or $R$ may be empty, and some of the odd components may contain only a single vertex.

Each \(s_j\) is an odd positive integer. In every maximum matching, at most one vertex from each component \(C_j\) is left unmatched. Therefore every admissible set \(T_i\) contains at most one vertex from each \(C_j\), and no other vertices.

We upper bound the number of sets in the antichain of the $T_i$'s. and this serves as an upper bound on $N$. We may assume that $|R|=0$, as otherwise, the vertices of $R$ can be inserted into one of the odd components. Likewise, we may assume that $|S|=0$. This is because otherwise, we can remove all vertices of $S$. This corresponds to an $(q,p')$ graph $G'$ with $p'$ (the maximum matching size) equal to $p - |S|$. Every \(T_i\) that is admissible in $G$ is also admissible in $G'$, so any upper bound on the antichain for $G'$ is an upper bound on that for $G$.

Thus, we may assume that $\sum_{i=1}^k s_i = q + p$ and $k = q - p$.

For fixed \(|T_i|=t\), the number of possible choices for \(T_i\) is
\[
e_t(s_1,\ldots,s_k) = \sum_{1 \le a_1 \le a_2 \ldots \le a_t \le k} s_{a_1}s_{a_2}\ldots s_{a_t},
\]
where \(e_t\) denotes the \(t\)-th elementary symmetric polynomial.

Let $\gamma = \frac{1}{k}\sum_{i=1}^k s_i$ be the average size of an odd component. For fixed \(\gamma\), the expression \(e_t(s_1,\ldots,s_k)\) is maximized
when the odd integers \(s_j\) are concentrated on two adjacent odd integers, $b$ and $b+2$, satisfying $b \le \gamma \le b+2$. (Suppose for the sake of contradiction that there are two components of odd sizes $b_1 \ge 1$ and $b_2 \ge b_1 + 4$. Then the number of ways of selecting one vertex from them is $b_1 + b_2$, and the number of ways of selecting two vertices is $b_1 b_2$. Changing their sizes to $b_1 + 2$ and $b_2 -2$,  the number of ways of selecting one vertex remains $b_1 + b_2$, but the number of ways of selecting two vertices increases to $(b_1+2)(b_2-2) = b_1 b_2 +2(b_2 - b_1) - 4 > b_1 b_2$.) 

Suppose first that $p \le \frac{q}{2}$, implying that $\gamma \le 3$. Then there are $q-2p$ components of size~1, and $p$ components of size~3. From each component of size~1, $T_i$ can either select its vertex or not (two possibilities), and from each component of size~3, $T_i$ can either select any one of its vertices or no vertex (four possibilities). So the number of possible options for $T_i$ is at most $2^{q-2p} \cdot 4^p = 2^q$.




Suppose now that $p > \frac{q}{2}$. In this case, the proof of Claim~\ref{cl:alpha} gives an upper bound of $(1.98 + o(1))^q$.

This completes the proof when the size of the maximum matching in $G$ is $p$. We now consider the remaining case, in which the maximum matching in $G$ has size $p+1$. The analysis for this case is also based on the Gallai-Edmonds decomposition of the graph, but differs from the previous case in the following aspect. When the maximum matching has size $p$, then $T_i$ is not allowed to contain any vertex from $S \cup R$, and can contain at most one vertex from each odd component. When the maximum matching has size $p+1$, additional vertices become eligible for $T_i$. Specifically, $T_i$ may contain one or two vertices from $S \cup R$, or alternatively, one or two vertices (beyond one that is already allowed) from some odd component (as the maximum matching in the graph that remains has to be of size $p$, not $p+1$, and not $p-1$). Consequently, the upper bound on $\alpha(q)$ increases by a factor of at most $O(q^2)$, a term hidden by the additional $o(1)$ term in the base of the exponent in our upper bound.
\end{proof}

\section{Rainbow matchings in hypergraphs}
\label{sec:rainbow}

In this section, hypergraphs in which the minimum hyperedge size is~2 (as in all previous sections in the paper) are denoted by $G$, whereas hypergraphs in which a hyperedge may also contain only a single vertex are denoted by $H$.

In our correspondence between colored hypergraphs and allocation problems, a vertex represents an item. Each agent $i$ corresponds to a {\em color} $i$, and contributes to the hypergraph hyperedges that correspond to bundles acceptable by $i$ (of value at least $MMS(\items, v_i, n)$), colored by color $i$. A hyperedge may have several copies in the hypergraph, each with a different color (as a bundle may be acceptable for several agents).

A matching in a hypergraph is {\em perfect} if the union of its hyperedges contains all vertices of the hypergraph. For integer $p$, a $p$-matching is a matching with exactly $p$ hyperedges. A {\em rainbow matching} is a matching in which no two hyperedges have the same color.

Consider first an allocation instance with arbitrary valuations. In this case, $\mu(n) < n$ (by Theorem~\ref{thm:main}), and so $m < n$. Consider the hypergraph $H(V, E)$ with $V = \items$. For every agent $i$, include in $E$ those hyperedges that correspond to the nonempty sets in the $MMS(\items, v_i, n)$ partition, and color them by color $i$. Determining $\mu(n)$ for arbitrary valuations can be seen to be equivalent to the following problem.

\begin{question}
\label{qu:arbitrary}
    For a given $n$, what is the largest value of $m$ for which the following holds? Every hypergraph $H$ on $m$ vertices that has at least $n$ perfect matchings, each in a different color, has a perfect rainbow matching.
\end{question}

Our Theorem~\ref{thm:main} shows that the answer to question~\ref{qu:arbitrary} is $m = \log(n) \pm O(1)$.

Moving to monotone valuations over goods, $\mu(n) > n$, and the MMS partition of every agent does not contain empty parts.  Determining $\mu(n)$ for monotone valuations over goods can be seen to be equivalent to the following problem.

\begin{question}
\label{qu:goods}
    For a given $n$, what is the largest value of $m$ for which the following holds? Every hypergraph $H$ on $m$ vertices that has at least $n$ perfect $n$-matchings, each in a different color, has a rainbow $n$-matching.
\end{question}

Note that in Question~\ref{qu:goods}, the rainbow matching need not be perfect. This is because goods that remain can be allocated arbitrarily, preserving the fact that each agent got at least her MMS.

Our Theorem~\ref{thm:main} shows that the answer to question~\ref{qu:goods} is $m = n + \Theta(\log(n))$.

To get tighter bounds on the constant hidden in the $\Theta$ notation, our Lemma~\ref{lem:ncApplication} reduces Question~\ref{qu:goods} to the following question.

\begin{question}
\label{qu:reduction}
    For an integer $c$, what is the smallest value of $n$ for which the following holds for every $1 \le d \le c$? Every hypergraph $G$ on $c+d$ vertices that has at least $n$ $d$-matchings, each in a different color, has a rainbow $d$-matching.
\end{question}

By our naming convention, in the hypergraph $G$ of Question~\ref{qu:reduction}, every hyperedge has size at least~2. Theorem~\ref{thm:submodular} implies that $n \ge (2 -  o(1))^c$, whereas Lemma~\ref{lem:ncUpperBound} implies that $n \le (2.0922 + o(1))^c$.

In the context of chores, it is convenient to introduce a notion of a {\em rainbow cover}. A set of hyperedges forms a rainbow cover if they all have distinct colors and their union contains all vertices.  In a cover, hyperedges are not required to be disjoint. Determining $\mu(n)$ for monotone valuations over chores can be seen to be equivalent to the following problem.

\begin{question}
\label{qu:chores}
    For a given $n$, what is the largest value of $m$ for which the following holds? Every hypergraph $H$ on $m$ vertices that has at least $n$ perfect $n$-matchings, each in a different color, has a rainbow cover.
\end{question}

Theorem~\ref{thm:main} implies that the answer to Question~\ref{qu:chores} is $n + \Theta(\log n)$.

\section{Additive valuations}
\label{sec:additive}

In this section we show that $\mu(n)$ in the case of additive valuations over goods is higher than in the case of arbitrary monotone valuations over goods. 

Similar to~\cite{hummel2023lowerboundsmaximinshare}, we may assume that the allocation instance has additive valuations with {\em same order preferences} (SOP) over the goods, meaning that for every valuation function $v_i$, and for every two items $e_j$ and $e_k$ with $j \le k$, $v_i(e_j) \ge v_i(e_k)$. 

For SOP instances with $n + c$ items, for $c < n$, each of the top $n-c$ items forms a part in the $MMS(\items,v_i,n)$ partition of each agent. Consequently, Lemma~\ref{lem:ncApplication} can easily be replaced by the following lemma.  

\begin{lemma}
    \label{lem:ncApplicationAdditive}
    For every integer $c \ge 1$ and $n_c$ as in Definition~\ref{def:nc}, for every $n \ge n_c + c$, every instance with additive valuations and $n + c$ goods has an MMS allocation.
\end{lemma}

Our improvement for $\mu(n)$ is an immediate consequence of the following improved bound, compared to the one in Lemma~\ref{lem:ncUpperBound}

\begin{lemma}
    \label{lem:ncUpperBoundAdditive}
    For additive valuations over goods, the value of $n_c$ in definition~\ref{def:nc} satisfies $n_c \le (1.89 + o(1))^c$.
\end{lemma}

To prove Lemma~\ref{lem:ncUpperBoundAdditive}, we modify the proof of Lemma~\ref{lem:alpha}.  Recall $\alpha(q)$ from Definition~\ref{def:alpha}, and that Lemma~\ref{lem:hypergraphReduction} states that $n_c \le \alpha(q) + (1.91 + o(1))^c$, for $q = c+2$. One modification is to strengthen the notion of anti-chain used in the third part of Definition~\ref{def:alpha}. 

\begin{proposition}
\label{pro:enhancedAntiChain}
For additive valuations, one can add to Definition~\ref{def:alpha} the condition that the sets $S_i \cup T_i$ also form an anti-chain. Namely, there are no two free pairs $(S_i,T_i)$ and $(S_j,T_j)$ for which $S_j \cup T_j \subseteq S_i \cup T_i$.
\end{proposition}
    
For additive valuations, a violation of Proposition~\ref{pro:enhancedAntiChain} implies that $$\max[v_j(S_i), v_j(T_i)] \ge \min[v_j(S_j), v_j(T_j)].$$
The proof of Lemma~\ref{lem:hypergraphReduction} can easily be adapted to imply Proposition~\ref{pro:enhancedAntiChain}. This is because if there is a partition $Q = (T_1, \ldots T_d)$ for agent $i$ that violates Proposition~\ref{pro:enhancedAntiChain} with agent $j$, then agent $j$ necessarily finds either $T_1$ or $T_2$ acceptable, and the partition need not be placed in the class $C_2$.

An additional modification is the following. Recall that in the hypergraph $G(V,E)$ of Definition~\ref{def:alpha}, there are $q + p$ vertices, and each vertex corresponds to an item. {We assume throughout that $p \ge 3$, because for smaller values of $p$, the proofs are simpler and the bounds are stronger (for additive valuations, MMS allocations for $p \le 2$ agents always exist).} As the allocation instance can be assumed to satisfy SOP, there is a set $U\subset V$ of size $|U| = p$ that contains the top $p$ items in the SOP preference order. 

\begin{proposition}
    \label{pro:additiveU}
    For additive valuations and $U$ as above, one can add to Definition~\ref{def:alpha} the condition that $|(S_i \cup T_i) \cap U| \le 2\sqrt{q}$ for every free pair $(S_i,T_i)$.
\end{proposition}

\begin{proof}
    Assume for the sake of contradiction that $|(S_i \cup T_i) \cap U| > 2\sqrt{q}$. 
    By definition, each partition (and hence also partition $i$) has $p$ parts. By definition, $|U| = p$. Hence, on average, each part of partition $i$ has one item from $U$. For each additional item that a part has from $U$, there must be some part that has no item. This shows that under the assumption, at least $2\sqrt{q}$ parts have no item from $U$. {(More accurately, the number of these parts is at least $\lfloor 2\sqrt{q} \rfloor$, but to keep the presentation simple, we omit the floor notation.)} 
    Each such part must contain at least $\sqrt{q}$ items from $V \setminus U$. This is because otherwise, we can drop an item from either $S_i$ or $T_i$ and still keep the respective part acceptable (at least one of these sets has at least $\sqrt{q} + 1$ items from $U$, and every such item has value at least as large as any item in $V \setminus U$). Thus the total number of items has to be at least $2\sqrt{q} \cdot \sqrt{q} + |S_i| + |T_i| > 2q \ge p+q$, which is a contradiction.
\end{proof}

Using the above two propositions, we can modify the proof of Claim~\ref{cl:alpha} and prove the following claim.

\begin{claim}
\label{cl:alphaAdditive}
For additive valuations,
\[
\alpha(q)\le (1.89+o(1))^q.
\]
\end{claim}

\begin{proof}
Let $G(V,E)$ be an $(q,p)$ hypergraph and let $C$ be an anti-chain of free
pairs (as in Definition~\ref{def:alpha}).

For every free pair $(S_i,T_i)\in C$, define $R_i = V\setminus (S_i\cup T_i)$.

By Proposition~\ref{pro:enhancedAntiChain}, the families
$\{S_i\}$, $\{T_i\}$, $\{R_i\}$,
are all anti-chains.

Let $U$ be the set appearing in Proposition~\ref{pro:additiveU}. By that
proposition,
\[
|(S_i\cup T_i)\cap U|=O(\sqrt q).
\]
Consequently, each of the sets
$S_i\cap U$, $T_i\cap U$, 
has size $O(\sqrt q)$.

Hence, the total number of possible triples $(S_i\cap U,\;T_i\cap U,\;R_i\cap U)$ is at most $2^{O(\sqrt q\log q)}$.

Fix one such triple. Since $(S_i,\ T_i,\ R_i)$ partition $V$, fixing two of the three intersections with $U$ determines the
third. Therefore, after fixing the triple, inclusion relations among the sets
$S_i,T_i,R_i$ are completely determined by their intersections with
$V\setminus U$.

Define
\[
S_i' = S_i\setminus U,\qquad
T_i' = T_i\setminus U,\qquad
R_i' = R_i\setminus U.
\]

The sets $(S_i',\ T_i',\ R_i')$ partition $V\setminus U$, and $|V\setminus U|=q$.

Therefore
\[
\min\{|S_i'|,|T_i'|,|R_i'|\}\le \frac q3 .
\]

Partition the family of free pairs into three classes according to which of
$S_i',T_i',R_i'$ attains the minimum size (breaking ties arbitrarily). One of
the three classes contains at least a third of the free pairs.

Fix such a class, and let $X_i'$ denote the corresponding smallest set. Then
$|X_i'|\le \frac q3$ for every member of the class. Moreover, because the corresponding original
family ($\{S_i\}$, $\{T_i\}$, or $\{R_i\}$) is an anti-chain and the
intersections with $U$ have been fixed, the sets $X_i'$ form an anti-chain in $V\setminus U$, {and are all distinct.} Hence, 

\[
|\{X_i'\}|
\le
\sum_{j=0}^{\lfloor q/3\rfloor}
\binom qj .
\]

Using the standard entropy estimate,

\[
\sum_{j=0}^{\lfloor q/3\rfloor}
\binom qj
=
\left(
2^{H(1/3)}+o(1)
\right)^q,
\]

where \(H\) is the binary entropy function. Since $2^{H(1/3)}
=
1.889881\ldots$, we obtain

\[
|\{X_i'\}|
\le
(1.89+o(1))^q.
\]

Multiplying by the three possible choices of class and by the
\(2^{O(\sqrt q\log q)}\) possibilities for the intersections with \(U\), we get

\[
\alpha(q)
\le
3\cdot
2^{O(\sqrt q\log q)}
\cdot
(1.89+o(1))^q
=
(1.89+o(1))^q.
\]

This completes the proof.
\end{proof}

Putting all parts together as in Appendix~\ref{sec:goods}, we get that for additive valuations over goods, $\mu(n) \ge n + \frac{1}{\log 1.89}\log n - o(\log n) \ge n + (1.088 - o(1))\log n$.

\end{appendix}

\end{document}